\documentclass[11pt]{article}

\usepackage{hyperref}
\usepackage{fullpage}
\usepackage{amsmath}
\usepackage{amssymb}
\usepackage{latexsym}

\newtheorem{theorem}{Theorem}
\newtheorem{lemma}{Lemma}
\newtheorem{claim}{Claim}

\newenvironment{proof}{\smallskip\par\noindent {\em Proof. }}{~\hfill$\Box$~}

\title{Approximating 1-dimensional TSP requires $\Omega(n\log n)$ comparisons}
\begin{document}
\maketitle
\begin{abstract}
  We give a short proof that any comparison-based 
  $(n^{1-\epsilon})$-approximation algorithm
  for the 1-dimensional Traveling Salesman Problem (TSP) requires 
  $\Omega(n\log n)$ comparisons.
\end{abstract}

\section*{Introduction}
In September of 2012 David Eppstein asked the following question on 
\href{http://cstheory.stackexchange.com/questions/12593/approximate-1d-tsp-with-linear-comparisons}{cstheory.stackexchange.com}:

\begin{quote}
  ``The one-dimensional traveling salesperson path problem is, obviously, the same thing as sorting, and so can be solved exactly by comparisons in $O(n\log n)$ time, but it is formulated in such a way that approximation as well as exact solution makes sense. In a model of computation in which the inputs are real numbers, and rounding to integers is possible, it's easy to approximate to within a $1+O(n^{-c})$ factor, for any constant $c$, in time $O(n)$: find the min and max, round everything to a number within distance $(\max-\min)n^{-(c+1)}$ of its original value, and then use radix sort. But models with rounding have problematic complexity theory and this led me to wonder, what about weaker models of computation?

So, how accurately can the one-dimensional TSP be approximated, in a linear comparison tree model of computation (each comparison node tests the sign of a linear function of the input values), by an algorithm whose time complexity is $o(n\log n)$? The same rounding method allows any approximation ratio of the form $n^{1-o(1)}$ to be achieved (by using binary searches to do the rounding, and rounding much more coarsely to make it fast enough). But is it possible to achieve even an approximation ratio like $O(n^{1-\epsilon})$ for some $\epsilon>0$?''
\hfill \cite{eppstein2012}
\end{quote}

The purpose of this short note is to document, in a citable form, a negative answer.

\begin{theorem}\label{thm:1}
For any constant $\epsilon>0$, every comparison-based 
$n^{1-\epsilon}$-approximation algorithm for 1-dimensional TSP
requires $\Omega(\epsilon n\log n)$ comparisons in the worst case.
\end{theorem}
A ``comparison-based'' algorithm is one that 
queries the input only with binary (True/False) queries.

Note: after the initial draft of this note, it was pointed out to me
that Das, Kapoor and Smid proved the following result in 1997:
\begin{theorem}[\cite{das1997complexity}, Thm.~1]
Let $d\ge 1$ be an integer constant.  In the algebraic computation-tree model,
any algorithm that, given a set $S$ of $n$ points in $\mathbb{R}^d$ and a sufficiently
large real number $r<n$, computes an $r$-approximate TSP-tour for $S$,
takes $\Omega(n\log(n/r))$ time in the worst case.
\end{theorem}
They also give related results for minimum-spanning trees.

Note that, by taking $r=n^{1-\epsilon}$, their result
essentially implies Thm.~\ref{thm:1}.  
The proof in \cite{das1997complexity} uses a Ben-Or-style proof,
different from the direct proof here.

\section*{Proof of Theorem~\ref{thm:1}}
Fix any constant $\epsilon>0$ and an arbitrarily large positive integer $n$.

Consider just the $n!$ "permutation" input instances 
$(x_1,x_2,\ldots,x_n)$
that are permutations of $[n] = \{1,2,\ldots,n\}$.
The optimum solution for any such instance has cost $n-1$.

Define the {\em cost} of a permutation $\pi$
to be $c(\pi) = \sum_i |\pi(i+1) - \pi(i)|$.
Model the algorithm as taking as input a permutation $\pi$,
returning some permutation $\pi'$,
and paying cost $d(\pi,\pi') = c(\pi'\circ \pi)$.

Define $C=C(n,\epsilon)$ to be the minimum number of comparisons
for any comparison-based algorithm to achieve
competitive ratio $n^{1-\epsilon}$ on these instances.
Since opt is $n-1$, the algorithm must guarantee cost at most $n^{2-\epsilon}$.
We will show $C\ge \Omega(\epsilon n\log n) - O(n)$.

Define $P=P(n,\epsilon)$ to be, for any possible output $\pi'$,
the fraction of possible inputs for which output $\pi'$ 
would achieve cost at most $n^{2-\epsilon}$.
This fraction is independent of $\pi'$.

$P$ also equals the probability that, for a random permutation $\pi$,
its cost $c(\pi)$ is at most $n^{2-\epsilon}$.
(To see why, take $\pi'$ to be the identity permutation $I$.
Then $P$ is the fraction of inputs for which
$d(\pi,I)$ at most $n^{2-\epsilon}$,
but $d(\pi,I) = c(\pi)$.)

\begin{lemma}\label{lemma:1}
$C(n,\epsilon) \ge \log_2 1/P(n,\epsilon)$.
\end{lemma}
\begin{proof}
  Fix any algorithm that uses less than $\log_2 1/P$
  comparisons on each of the inputs.  The decision tree for the algorithm has depth less
  than $\log_2 1/P$, so there are less than $1/P$ leaves, and, for
  some output permutation $\pi'$, the algorithm gives $\pi'$ as output
  for more than a $P$ fraction of the inputs.  By definition of $P$,
  for at least one such input, the output $\pi'$ gives cost more than
  $n^{2-\epsilon}$, so the algorithm is not a $n^{1-\epsilon}$-approximation 
  algorithm.
\end{proof}

\begin{lemma}\label{lemma:2}
$P(n,\epsilon) \le \exp(-\Omega(\epsilon n\log n))$.
\end{lemma}

Before we prove Lemma~\ref{lemma:2}, note that the two lemmas together
prove the theorem: 
$$C
~\ge~ \log_2 \frac{1}{P}
~=~ \log_2 \exp(\Omega(\epsilon n\log n))
~=~ \Omega(\epsilon n\log n).$$

To finish proving the theorem, we prove Lemma~\ref{lemma:2}.

Let $\pi$ be a random permutation.
Recall that $P$ equals the probability that
its cost $c(\pi)$ is at most $n^{2-\epsilon}$.
Say that any pair $(i,i+1)$ is an {\em edge}
with cost $|\pi(i+1)-\pi(i)|$,
so $c(\pi)$ is the sum of the edge costs.

Suppose $c(\pi) \le n^{2-\epsilon}$.

Then, for any $q>0$, at most $n^{2-\epsilon}/q$ of the edges have cost $q$ or more.
Say that edges of cost less than $q$ are {\em cheap}.

Fix $q=n^{1-\epsilon/2}$.  Substituting and simplifying, at most $n^{1-\epsilon/2}$ of the 
edges are not cheap.

Thus, {\em at least} $n - n^{1-\epsilon/2} \ge n/2$ of the edges are cheap.
Thus, there is a set $S$ containing $n/2$ cheap edges.

\begin{claim}
For any given set $S$ of $n/2$ edges,
the probability that all edges in $S$ are cheap
is at most $\exp(-\Omega(\epsilon n \log n))$.
\end{claim}

Before we prove the claim, note that it implies Lemma~\ref{lemma:2} as follows.
By the claim, and the naive union bound,
the probability that any there {\em exists} such a set $S$ 
is at most 
$${n\choose n/2} \exp(-\Omega(\epsilon n \log n))
~\le~ 2^n \exp(-\Omega(\epsilon n \log n))$$
$$~\le~ \exp(O(n) -\Omega(\epsilon n \log n))
~\le~ \exp(-\Omega(\epsilon n \log n)).$$

To complete the proof of the lemma and the theorem,
we prove the claim.

Choose $\pi$ by the following process.
Choose $\pi(1)$ uniformly from $[n]$,
then choose $\pi(2)$ uniformly from $[n] - \{\pi(1)\}$,
then choose $\pi(3)$ uniformly from $[n]-\{\pi(1),\pi(2)\}$, etc.

Consider any edge $(i,i+1)$ in $S$.
Consider the time just after $\pi(i)$ has been chosen,
when $\pi(i+1)$ is about to be chosen.
Regardless of the first $i$ choices (for $\pi(j)$ for $j\le i$),
there are at least $n-i$ choices for $\pi(i+1)$,
and at most $2n^{1-\epsilon/2}$ of those choices will give the edge $(i,i+1)$
cost less than $n^{1-\epsilon/2}$ (making it cheap).

Thus, conditioned on the first $i$ choices,
the probability that the edge is cheap
is at most $\min(1,\frac{2n^{1-\epsilon/2}}{n-i})$.
Thus, the probability that all $n/2$ 
edges in $S$ are cheap is at most 
$$\prod_{(i,i+1)\in S} \min\Big(1,\frac{2n^{1-\epsilon/2}}{n-i}\Big).$$
Since $|S|\ge n/2$, there are at least $n/4$ edges in $S$
with $n-i\ge n/4$. Thus, this product is at most
$$\big(\frac{2n^{1-\epsilon/2}}{n/4}\big)^{n/4}
~\le~(8n^{-\epsilon/2})^{n/4}
~=~\exp(O(n)-\Omega(\epsilon n \log n))
~=~\exp(-\Omega(\epsilon n \log n)).$$

This proves the theorem.~\hfill$\Box$

\bibliographystyle{plain}
\bibliography{bib}

\end{document}